\documentclass{llncs}%

\usepackage[noend]{algpseudocode}

\usepackage{amsfonts}
\usepackage{amsmath}
\usepackage{amssymb}
\usepackage{graphicx}%

\usepackage{todonotes}
\usepackage{algorithm}
\usepackage[noend]{algpseudocode}

\usepackage{xcolor}

\renewcommand{\L}{\Delta}

\newcommand{\Z}{\mathcal{B}}
\newcommand{\X}{\mathcal{X}}

\newtheorem{nclaim}{Claim}

\newcommand{\dupsto}{\Rightarrow}
\newcommand{\contractsto}{\rightarrowtail}
\newcommand{\dupstok}[1]{\dupsto_{#1}}
\newcommand{\contractstok}[1]{\contractsto_{#1}}

\title{The Tandem Duplication Distance is NP-hard}

\author{Manuel Lafond\inst{1} \and Binhai Zhu\inst{2} \and Peng Zou\inst{2}}

\institute{Department of Computer Science, Universite de Sherbrooke, Sherbrooke, Quebec J1K 2R1, Canada, \texttt{manuel.lafond@usherbrooke.ca} \and
Gianforte School of Computing, Montana State University, Bozeman, MT 59717, USA, \texttt{bhz@montana.edu, peng.zou@student.montana.edu}}

\authorrunning{Lafond, Zhu and Zou.}
\begin{document}
\maketitle  
\begin{abstract}
In computational biology, tandem duplication is an important biological phenomenon which can occur either at the genome or at the DNA level.  A tandem duplication takes a copy of a genome segment and inserts it right after the segment  --- this can be represented as the string operation $AXB \dupsto AXXB$.  %For example, 
Tandem exon duplications have been found in many species such as human, fly or worm, and have been largely studied in computational biology. 

The \emph{Tandem Duplication} (TD) distance problem we investigate in this paper is defined as follows: given two strings $S$ and $T$ over the same alphabet, compute the smallest sequence of tandem duplications required to convert $S$ to $T$. The natural question of whether the TD distance can be computed in polynomial time was posed in 2004 by Leupold et al. and had remained open, despite the fact that tandem duplications have received much attention ever since. In this paper, we prove that this problem is NP-hard.  We further show that this hardness holds even if all characters of $S$ are distinct.  This is known as the \emph{exemplar} TD distance, which is of special relevance in bioinformatics.  One of the tools we develop for the reduction is a new problem called the \emph{Cost-Effective Subgraph}, for which we obtain W[1]-hardness results that might be of independent interest.
We finally show that computing the exemplar TD distance between $S$ and $T$ is fixed-parameter tractable.  
Our results open the door to many other questions, and we conclude with several open problems.
%Due to the flexibility of tandem duplications, we 
%define a new NP-hard problem called Cost-Effective Subgraph (CES) and then reduce CES to Tandem Duplication Distance. 
%\colorbox{cyan}{NOTE ML: Say more about FPT results here.  Also say something about exemplar.}
%To complete the arguments in the reverse direction, we duplicate the gadgets corresponding to the input graph many times and show that a sufficient number of
%them contracted (which is the inverse to tandem duplication operations) toward a solution for CES suffice. The technique might be useful for other hard combinatorial problems. Finally, we show that the Tandem Duplication Distance problem is FPT.   
\end{abstract}

\section{Introduction}
Tandem duplication is a biological process that creates consecutive copies of a segment of a genome during DNA replication.  Representing genomes as strings, this event transforms a string $AXB$ into another string $AXXB$.  This process is known to occur either at small scale at the nucleotide level, or at large scale at the genome level \cite{Charlesworth94,Rao06,Chen09,Letunic02,Sharp05}. For instance, it is known that the Huntington disease is associated with the duplication of 3 nucleotides {\tt CAG} \cite{Macdonald93}, whereas at genome level, tandem duplications are known to involve multiple genes during cancer progression \cite{Oesper12}.  Furthermore, gene duplication is believed to be the main driving force behind evolution, and the majority of duplications affecting organisms are believed to be of the tandem type (see e.g.~\cite{uneqcrnature}).

For these reasons, tandem duplications have received significant attention in the last decades, both in practice and theory. 
%the fields of combinatorics and formal languages.
%--- even though in reality tandem duplications usually occur with some other
%operations, like segment deletions. 
The combinatorial aspects of tandem duplications have been studied extensively by computational biologists~\cite{landau01,gascuel2003,Gusfield04,tremblay2011} and, in parallel, by various formal language communities~\cite{Dassow99, Wang00, Leupold04}.  From the latter perspective, a natural question arises: given a string $S$, what is the language that can be obtained starting from $S$ and applying (any number of) tandem duplications, i.e. rules of the form $AXB \rightarrow AXXB$?  This question was first asked in 1984 in the context of so-called \emph{copying systems}~\cite{Ehren84}.  Combined with results from~\cite{Bovet92}, it was shown that this language is regular if $S$ is on a binary alphabet, but not regular for larger alphabets.  These results were rediscovered 15 years later in~\cite{Dassow99,Wang00}.
%In particular in \cite{Dassow99}, the language obtained by iteratively duplicating
%any substring $X$ (i.e., with a rule $X\rightarrow XX$) was studied. It was shown that if the alphabet is
%binary, then starting with any sequence the resulting language is regular. However, if the alphabet has
%size three, then such a language is non-regular \cite{Wang00}. 
In \cite{Leupold04}, it was shown that
given an unbounded duplication language (`unbounded' means that the size $|X|$ of the duplicated string $X$ is not necessarily 
bounded by any constant), the membership, inclusion and regularity testing problems can all be decided in
linear time; same with the equivalence testing between two such languages. In \cite{Leupold04,Leupold05,Ito06},
similar problems are also considered when the duplication size $|X|$ is bounded.
%In~\cite{Leupold05}, the case of uniform, fixed duplication lengths is studied, leading to polynomial-time decidability. 
More recently in~\cite{Hassan16,Jain17}, the authors study the \emph{expressive power} of tandem duplications, a notion based on the subsequences that can be obtained from a copy operation.

In this work, we are interested in a question posed in \cite{Leupold04} (pp. 306, Open Problem 3) by Leupold et al., who raised the problem of computing the minimum number of tandem duplications to transform a string $S$ to another string $T$.  We call this the \emph{Tandem Duplication (TD) distance} problem.  
The TD distance is one of the many ways of comparing two genomes represented as strings in computational biology --- other notable examples include 
breakpoint~\cite{hannen96} and transpositions distances, the latter having recently been shown NP-hard in a celebrated paper of Bulteau et al.~\cite{bulteau12}. The TD distance has itself received 
special attention recently, owing to its role in cancer evolution~\cite{Letu18}.
%In~\cite{Letu18}, the authors present a related algorithmic question: what is the minimum number of tandem duplications required to obtain a target \emph{Parikh vector} (the vector of count of every character in the alphabet), which is shown NP-hard.

%The main contribution of this
%paper is that this problem is in fact NP-hard when the alphabet is unbounded. We also show that this problem
%is FPT (Fixed-Parameter Tractable).

\vspace{3mm}

\noindent
\textbf{Our results.} In this paper, we solve the problem posed by Leupold et al. in 2004 and show that computing the TD distance from a string $S$ to a string $T$ is NP-hard.  
We show that this result holds even if $S$ is \emph{exemplar}, i.e. if each character of $S$ is distinct.  Exemplar strings are commonly studied in computational biology~\cite{sankoff01}, since they represent genomes that existed prior to duplication events.  We note that simply deciding if $S$ can  be transformed into $T$ by a sequence of TDs still has unknown complexity.  In our case, we show that the hardness of minimizing TDs holds on instances in which such a sequence is guaranteed to exist.

As demonstrated by the transpositions distance in~\cite{bulteau12}, obtaining NP-hardness results for string distances can sometimes be an involving task.
Our hardness reduction is also quite technical, and one of the tools we develop for it is a new problem we call the \emph{Cost-Effective Subgraph}.  
In this problem, we are given a graph $G$ with a cost $c$, and we must choose a subset $X$ of $V(G)$.  Each edge with both endpoints in $X$ has a cost of $|X|$, every other edge costs $c$, and the goal is to find a subset $X$ of minimum cost. We show that this problem is W[1]-hard for parameter $p + c$, where $p$ is the cost that we can save below the upper bound $c|E(G)|$\footnote{In other words, if we were to state the maximization version of the Cost-Effective Subgraph problem, $p$ would be the value to maximize.  The minimization version, however, is more convenient to use for our needs.}. 
The problem enforces optimizing the tradeoff between covering many edges versus having a large subset of high cost, which might be applicable to other problems.  In our case it captures the main difficulty in computing TD distances.
We then obtain some positive results by showing that if $S$ is exemplar, then one can decide if $S$ can be transformed into $T$ using at most $k$ duplications in time 
$2^{O(k^2)} + poly(n)$.  The result is obtained through an exponential size kernel.
Finally, we conclude with several open problems that might be of interest to the theoretical computer science community.

%The difficulty of the NP-hardness proof lies in that there is no constraint on the duplications, i.e., it could occur anywhere. To handle this problem, we define a new problem called Cost-Effective Subgraph
%(CES), where all edges in the graph are accounted though with different weights depending on whether
%the vertices belong to the target set, and show that it is NP-hard. We then design the gadgets
%carefully so that when we have enough such gadgets, after certain steps, these gadgets can only be
%contracted (which is the dual of tandem duplication) in one of the two ways, leading a much higher
%cost for the corresponding CES problem. This technique could be used in some other related problems
%when the genomic operations have a high degree of freedom. Finally, we show that the Tandem Duplication
%Distance is FPT using some standard reduction rules.

This paper is organized as follows. In Section 2, we give basic definitions. In Section 3, we show that computing the
TD distance is NP-hard through the Cost-Effective Subgraph problem. In Section 4, we show that computing the exemplar TD distance is FPT.
In Section 5, we conclude the paper with several open problems.

\section{Preliminary notions}

%Let $T$ be a string of length $n$.  A \emph{contraction} of $T$ consists of taking two equal consecutive substrings of $T$ and removing one of them.   More formally, let $i,j$ be such that 
%$T[i .. i + j] = T[i + j + 1 .. i + 2j + 1]$.
%The \emph{contraction} of $T$ at $(i, i + j)$
%results in the string $T' = T[1 .. i + j]T[i + 2j + 2 .. n]$.  We may write $T \Rightarrow T'$ if $T'$ can be obtained from $T$ by a sequence of contractions.

%Given two strings $S$ and $T$, the MIN-TANDEM-DUP problem asks for the minimum number of contractions to transform $T$ into $S$.

We borrow the string terminology and notation from~\cite{gusfieldbook}.
Unless stated otherwise, all the strings in the paper are on an alphabet denoted $\Sigma$.  
For a string $S$, we write $\Sigma(S)$ for 
the subset of characters of $\Sigma$ that have at least one occurrence in $S$.  A string $S$ is called 
\emph{exemplar} if $|S| = |\Sigma(S)|$, i.e. each character present in $S$ occurs only once.
A \emph{substring} of $S$ is a contiguous sequence of characters within $S$.  A \emph{prefix} (resp. \emph{suffix}) is a substring that occurs at the beginning (resp. end) of $S$.  A \emph{subsequence} of $S$ is a string that can be obtained by successively deleting characters from $S$.

A \emph{tandem duplication} (TD) is an operation on a string $S$ that copies a substring $X$ of $S$ and inserts the copy after the occurrence of $X$ in $S$.  In other words, a TD transforms $S = AXB$ into $AXXB$.  Given another string $T$, we write $S \dupsto T$ if there exist strings $A, B, X$ such that $S = AXB$ and $T = AXXB$.
More generally, we write $S \dupstok{k} T$ if there exist $S_1, \ldots, S_{k - 1}$ such that $S \dupsto S_1 \dupsto \ldots \dupsto S_{k-1} \dupsto T$.  We also write $S \dupstok{*} T$ if there exists some $k$ such that $S \dupstok{k} T$.  

\begin{definition}
The TD distance $dist_{TD}(S, T)$ between two strings $S$ and $T$ is the minimum value of $k$ 
satisfying $S \dupstok{k} T$.  If $S \dupstok{*} T$ does not hold, then $dist_{TD}(S, T) = \infty$.
\end{definition}

A \emph{square string} is a string of the form $XX$, i.e. a concatenation of two identical substrings.
Given a string $S$, a \emph{contraction} is the reverse of a tandem duplication.  That is, it takes a square string $XX$ contained in $S$ and deletes one of the two copies of $X$.  
We write $T \contractsto S$ if there exist strings $A, B, X$ such that $T = AXXB$ and $S = AXB$.
We also define $T \contractstok{k} S$ and $T \contractstok{*} S$ for contractions analogously as for TDs (note that $T \contractstok{k} S$ if and only if $S \dupstok{k} T$ and $T \contractstok{*} S$ if and only if $S \dupstok{*} T$).  When there is no possible confusion, we will sometimes write $T \contractsto S$ instead $T \contractstok{*} S$.

We have the following problem.

\vspace{4mm}

\noindent
The \textsf{$k$-Tandem Duplication} ($k$-TD) problem: \\
\noindent
\textbf{Input}: two strings $S$ and $T$ over the same alphabet $\Sigma$ and an integer $k$. \\
\noindent
\textbf{Question}: is $dist_{TD}(S, T) \leq k$?

\vspace{4mm}

In the \textsf{Exemplar-$k$-TD} variant of this problem, 
$S$ is required to be exemplar.  
In either variant, we may call $S$ the \emph{source string} and $T$ the \emph{target string}. 
We will often use the fact that $S$ and $T$ form a YES instance if and only if $T$ can be transformed into $S$ 
by a sequence of at most $k$ contractions.  See Fig.\ref{f1} for a simple example.

%\colorbox{cyan}{NOTE ML: MAYBE PROVIDE FIGURE/EXAMPLE TO MAKE THE PAPER LESS DRY}
\begin{figure}[htbp]
\begin{eqnarray*}
Sequence && Operations
\\
Sequence\quad T=\langle a,c,\underline{g,g},a,c,g \rangle &~~~~~& contraction~on~ \langle g,g \rangle
\\
\langle \underline{a,c,g,a,c,g}\rangle && contraction~on~ \langle a,c,g,a,c,g \rangle
\\
Sequence\quad S=\langle a,c,g \rangle && ~~
\end{eqnarray*}
\caption{An example for transforming sequence $T$ to $S$ by two contractions.  The corresponding sequence of TDs from $S$ to $T$ would duplicate $a,c,g$, and then duplicate the first $g$.}
\label{f1}
\end{figure}

%In the ``promise'' version of the problem, we assume that a sequence of contractions transforming $T$ to $S$ exists.  
We recall that although we study the minimization problem here, it is unknown whether the question $S \dupstok{*} T$ can be decided in polynomial time.
Nonetheless, our NP-hardness reduction applies to `promise' instances in which $S \dupstok{*} T$ always holds.

\section{NP-hardness of Exemplar-$k$-TD}

To facilitate the presentation of our hardness proof, we first make an intermediate reduction using the \textsf{Cost-Effective Subgraph} problem, which we will then reduce to the promise version of the \textsf{Exemplar-$k$-TD} problem.

\subsection*{The \textsf{Cost-Effective Subgraph} problem}

Suppose we are given a graph $G = (V, E)$ and an integer cost $c \in \mathbb{N}^{> 0}$.  For a subset $X \subseteq V$, let $E(X) = \{uv \in E : u, v \in X\}$ denote the edges inside of $X$.  The \emph{cost} of $X$ is defined as 

$$
cost(X) = c \cdot (|E(G)| - |E(X)|) + |X| \cdot |E(X)|
$$

The \textsf{Cost-Effective Subgraph} problem asks for a subset $X$ of minimum cost.  In the decision version of the problem, we are given an integer $k$ and we want to know if there is a subset $X$ whose cost is at most $k$.  Observe that $X = \emptyset$ or $X = V$ are possible solutions.

The idea is that each edge ``outside'' of $X$ costs $c$ and each edge ``inside'' costs $|X|$.  Therefore, we pay for each edge not included in $X$, but if $X$ gets too large, we pay more for edges in $X$.  We must therefore find a balance between the size of $X$ and its number of edges.  The connection with $k$-TD can be roughly described as follows: in our reduction, we will have many substrings which need to be deleted through contractions.  We will have to choose an initial set of contractions $X$ and then, each  substring will have two ways to be contracted: one of cost $c$, and the other of cost $X$. 

An obvious solution for a \textsf{Cost-Effective Subgraph} is to take $X = \emptyset$, which is of cost $c|E(G)|$.  Another formulation of the problem could be whether there is a subset $X$ of cost at most $c|E(G)| - p$, where $p$ can be seen as a ``profit'' to maximize.  Treating $p$ as a parameter, 
we show the NP-hardness and W[1]-hardness in parameters $c + p$ of the \textsf{Cost-Effective Subgraph} problem (we do not study the parameter $k$).  Our reduction to $k$-TD does not preserve W[1]-hardness and we only use the NP-hardness in this paper, but the W[1]-hardness might be of independent interest.  

Before proceeding, we briefly argue the relevance of parameter $c$ in the W[1]-hardness.  If $c$ is a fixed constant, then we may assume that any solution $X$ satisfies $|X| \leq c$.  This is because if $|X| > c$, every edge included in $X$ will cost more than $c$ and putting $X = \emptyset$ yields a lower cost.  Thus for fixed $c$, it suffices to brute-force every subset $X$ of size at most $c$ and we get a $n^{O(c)}$ time algorithm.  Our W[1]-hardness shows that it is difficult to remove this exponential dependence between $n$ and $c$.

\begin{theorem}
The \textsf{Cost-Effective Subgraph} problem is NP-hard and W[1]-hard for parameter $c + p$.
\end{theorem}

\begin{proof}
We reduce from \textsf{CLIQUE}, a classic NP-hard problem where we are given a graph $G$ and an integer $k$ and must decide whether $G$ contains a clique of size at least $k$.  The problem is also W[1]-hard in parameter $k$~\cite{fellows95}.  We will assume that $k$ is even (which does not alter either hardness results).

Let $(G, k)$ be a \textsf{CLIQUE} instance, letting $n := |V(G)|$ and $m := |E(G)|$.  The graph in our \textsf{Cost-Effective Subgraph} instance is also $G$.  We set the cost $c = 3k/2$, which is an integer since $k$ is even, and put 
\begin{align*}
    r &:= c \left( m - {k \choose 2}\right) + k{k \choose 2} = cm + {k \choose 2}(k - c) = cm - k/2{k \choose 2}
\end{align*}

We ask whether $G$ admits a subgraph $X$ satisfying $cost(X) \leq r$.
We show that $(G, k)$ is a YES instance to \textsf{CLIQUE} if and only if $G$ contains a set $X \subseteq V(G)$ of cost at most $r$.  This will prove both NP-hardness and W[1]-hardness in $c + p$ (noting that here $p = k/2{k \choose 2}$).

The forward direction is easy to see. If $G$ is a YES instance, it has a clique $X$ of size (exactly) $k$.  Since $|E(X)| = {k \choose 2}$, the cost of $X$ is precisely $r$.

Let us consider the converse direction.  Assume that $(G, k)$ is a NO instance of \textsf{CLIQUE}.  Let $X \subseteq V(G)$ be any subset of vertices.  We will show that $cost(X) > r$. There are $3$ cases to consider depending on $|X|$.

\vspace{2mm}
\noindent
\emph{Case 1}: $|X| = k$.  Since $G$ is a NO instance, $X$ is not a clique and thus $|E(X)| = {k \choose 2} - h$, where $h > 0$.  
We have that $cost(X) = c(m - {k \choose 2} + h) + k({k \choose 2} - h) = cm + {k \choose 2}(k - c) + h(c - k) = r + h(c - k)$.
Since $c > k$ and $h > 0$, the cost of $X$ is strictly greater than $r$.

\vspace{2mm}
\noindent
\emph{Case 2}: $|X| = k + l$ for some $l > 0$.  Denote $|E(X)| = {{k+l} \choose 2} - h$, where $h \geq 0$ (actually, $h > 0$ but we do not bother).
The cost of $X$ is 
\begin{align*}
    cost(X) &= c\left(m - {{k+l} \choose 2} + h\right) + (k + l)\left({{k + l} \choose 2} - h\right) 
    \\&= cm + {{k + l} \choose 2}(k + l - c) + h(c - k - l )
    \\&= cm + {{k + l} \choose 2}(l - k/2) + h(k/2 - l)
\end{align*}

Consider the difference  
\begin{align*}
    cost(X) - r &= {{k + l} \choose 2}(l - k/2) - (-k/2){k \choose 2}  + h(k/2-l) 
    \\&=\frac{3 k l^2}{4} - \frac{k l}{4} + \frac{l^3}{2} - \frac{l^2}{2} + h(k/2-l) 
    %\\&=3 k l^2 + 3 k l + 2 l^3 - 2 l+ h(2k - 4l - 4) 
\end{align*}

If $k/2 - l \geq 0$, then the difference is clearly above $0$ regardless of $h$, and then $cost(X) > r$ as desired.  Thus we may assume that $k/2 - l < 0$.  In this case, we may assume that $h = {{k + l} \choose 2}$, as this minimizes $cost(X)$.
But in this case, $cost(X) = cm + {{k+l} \choose 2}(l - k/2) + {{k+l} \choose 2}(k/2 - l) = cm > r$.

\vspace{2mm}
\noindent
\emph{Case 3}: $|X| = k - l$, with $l > 0$.  If $k = l$, then $X = \emptyset$ and $cost(X) = cm > r$.  So we assume $k > l$.
%Since $G$ is a NO instance, we know that $X$ is not a clique (because the clique number of $G$ is at most $k/2$).
Put $|E(X)| = {{k-l} \choose 2} - h$, where $h \geq 0$.
We have 
\begin{align*}
cost(X) &= c\left( m - {{k-l} \choose 2} + h \right) + (k - l)\left( {{k-l} \choose 2} - h \right) \\
&= cm + {{k-l} \choose 2}(k - l - c) + h(c - k + l) \\
&= cm + {{k - l} \choose 2}(-k/2 - l) + h(k/2 + l)
\end{align*}

The difference with this cost and $r$ is 
\begin{align*}
cost(X) - r &= {{k - l} \choose 2}(-k/2 - l) - (-k/2){k \choose 2} + h(k/2 + l) \\
&= \frac{3 k l^2}{4} + \frac{k l}{4} - \frac{l^3}{2} - \frac{l^2}{2} + h(k/2+l) \\
&> \frac{1}{4}(3l^3 + l^2) - \frac{1}{2}(l^3 + l^2) \geq 0
\end{align*}

the latter since $k > l \geq 1$.  Again, it follows that $cost(X) > r$.
%\qed
\end{proof}

\subsection*{Reduction to Exemplar-$k$-TD}

Since the reduction is somewhat technical, we provide an overview of the techniques that we will use.  Let $(G, c, r)$ be a \textsf{Cost-Effective Subgraph} instance where $c$ is the cost and $r$ the optimization value, and with vertices $V(G) = \{v_1, \ldots, v_n\}$.  We will construct strings $S$ and $T$ and argue on the number of contractions to go from $T$ to $S$.
We would like our source string to be $S = x_1x_2 \ldots x_n$, where each $x_i$ is a distinct character that corresponds to vertex $v_i$.
Let $S'$ be obtained by doubling every $x_i$, i.e. 
$S' = x_1x_1x_2x_2 \ldots x_nx_n$.
Our goal is to put $T = S'E_1E_2 \ldots E_m$, where each $E_i$ is a substring gadget corresponding to edge $e_i \in E(G)$ that we must remove to go from $T$ to $S$.  
In a contraction sequence from $T$ to $S$, we make it so that we first want to contract some, but not necessarily all, of the doubled $x_i$'s of $S'$, resulting in another string $S''$.  
Let $t$ be the number of $x_i$'s contracted from $S'$ to $S''$.
For instance, we could have $S'' = x_1x_1x_2x_3x_3x_4x_5x_5$, where only $x_2$ and $x_4$ were contracted, and thus $t = 2$.  The idea is that these contracted $x_i$'s correspond to the vertices of a cost-effective subgraph.
After $T$ is transformed to $S''E_1 \ldots E_m$, we then force each $E_i$ to use $S''$ to contract it.
For $m = 3$, a contraction sequence that we would like to enforce would take the form
$$
\underline{S'}E_1E_2E_3 \contractsto \underline{S''E_1}E_2E_3 \contractsto \underline{S''E_2}E_3 \contractsto \underline{S''E_3} \contractsto \underline{S''} \contractsto S
$$

where we underline the substring affected by contractions at each step.
We make it so that when contracting $S''E_iE_{i+1} \ldots E_m$ into $S''E_{i+1} \ldots E_m$, we have two options.
Suppose that $v_j, v_k$ are the endpoints of edge $e_i$.  If, in $S''$, we had chosen to contract $x_j$ and $x_k$, we can contract $E_i$ using a sequence of $t$ moves.  
Otherwise, we must contract $E_i$ using another more costly sequence of $c$ moves.  The total cost to eliminate the $E_i$ gadgets will be $c(m - e) + te$, where $e$ is the number of edges that can be contracted using the first choice, i.e. for which both endpoints were chosen in $S''$.   

Unfortunately, constructing $S'$ and the $E_i$'s to implement the above idea is not straightforward.  The main difficulty lies in forcing an optimal solution to behave as we describe -- i.e. enforcing going from $S'$ to $S''$ first, enforcing the $E_i$'s to use $S''$, and enforcing the two options to contract $E_i$ with the desired costs.  In particular, we must replace the $x_i$'s by carefully constructed substrings $X_i$.  We must also repeat the sequence of $E_i$'s a certain number $p$ times.  We now proceed with the technical details.

\begin{theorem}\label{thm:main-np-hardness}
The \textsf{Exemplar-$k$-TD} problem is NP-complete.
\end{theorem}

\begin{proof}
To see that the problem is in NP, note that $dist_{TD}(S, T) \leq |T|$ since each contraction from $T$ to $S$ removes a character.  Thus a sequence of contractions can serve as a certificate, has polynomial size and is easy to verify.

For hardness, we reduce from the \textsf{Cost-Effective Subgraph} problem.  
Let $(G, c, r)$ be an instance of \textsf{Cost-Effective Subgraph}, letting $n := |V(G)|$ and $m := |E(G)|$.  Here $c$ is the ``outsider edge'' cost and we ask whether there is a subset $X \subseteq V(G)$ such that $c(m - |E(X)|) + |X||E(X)| \leq r$.  
We denote $V(G) = \{v_1, \ldots, v_n\}$ and $E(G) = \{e_1, \ldots, e_m\}$.  The ordering of vertices and edges is arbitrary but remains fixed for the remainder of the proof.
For convenience, we allow the edge indices to loop through $1$ to $m$, and so we put $e_i = e_{i + lm}$ for any integer $l \geq 0$.  Thus we may sometimes refer to an edge $e_k$ with an index $k > m$, meaning that $e_k$ is actually the edge $e_{((k - 1) \mod m) + 1}$.

\vspace{2mm}
\noindent
\textbf{The construction.}
Let us first make an observation.  If we take an exemplar string $X = x_1 \ldots x_l$ (i.e. a string in which no character occurs twice), we can double its characters and obtain a string $X' = x_1x_1 \ldots x_lx_l$.  The length of $X'$ is only twice that of $X$ and $dist_{TD}(X, X') = l$, i.e. going from $X'$ to $X$ requires $l$ contractions.  We will sometimes describe pairs of strings $X$ and $X'$ at distance $l$ without explicitly describing $X$ and $X'$, but the reader can assume that $X$ starts as an exemplar string and we obtain $X'$ by doubling it.

Now we show how to construct $S$ and $T$.  First let $d,p \in \mathbb{N}$ be large (but polynomial) integers.  We choose $p$ to be a multiple of $m$.  For concreteness, we put $d = m + 1$ and $p = m(n + m)^{10}$, but it is enough to think of these values as simply ``large enough''.  Instead of doubling $x_i$'s as in the intuition paragraph above, we will duplicate some characters $d$ times.  Moreover, we can't create a $T$ string that behaves exactly as described above, but we will show that we can append $p$ copies of carefully crafted substring to obtain the desired result.  We need $d$ and $p$ to be high enough so that ``enough'' copies behave as we desire.  

For each $i \in [n]$, define an exemplar string $X_i$ of length $d$.  Moreover, create enough characters so that no two $X_i$ string contain a character in common.  Let $X_i^d$ be a string satisfying $dist_{TD}(X_i, X_i^d) = d$.

Then for each $j \in \{0, 1, \ldots, 2p\}$, define an exemplar string $B_j$.  Ensure that no $B_j$ contains a character from an $X_i$ string, and no two $B_j$'s contain a common character.
The $B_j$ strings can consist of a single character, with the exception of $B_0$ and $B_1$ which are special.  
We assume that for $B_0$ and $B_1$, we have strings $B_0^{*}$ and $B_1^*$ such that
\begin{align*}
dist_{TD}(B_0, B_0^{*}) &= dc + 2d - 2 \\
dist_{TD}(B_1, B_1^{*}) &= dn + 2d - 1
\end{align*}
The $B_j$'s are the building blocks of larger strings.  For each $q \in [2p]$, define 
\begin{align*}
\Z_q &= B_qB_{q-1} \ldots B_2B_1B_0    
&\Z_q^0 &= B_qB_{q-1} \ldots B_2B_1B_0^{*} \\
\Z_q^1 &= B_qB_{q-1} \ldots B_2B_1^{*} B_0
& \Z_q^{01} &= B_qB_{q-1} \ldots B_2B_1^{*} B_0^{*}
\end{align*}
These strings are used as ``blockers'' and prevent certain contractions from happening.
Also define the strings
\begin{align*}
\X = X_1X_2 \ldots X_n && 
\X^d = X_1^dX_2^d \ldots X_n^d
\end{align*}
and for edge $e_q = v_iv_j$ with $q \in [p]$ whose endpoints are $v_i$ and $v_j$, define
\begin{align*}
\X_{e_q} = %\X_{v_iv_j} = 
X_1^d \ldots X_{i-1}^d X_i X_{i+1}^d \ldots 
X_{j-1}^d X_j X_{j+1}^d \ldots 
X_n^d
\end{align*}
Thus in $\X_{e_q}$, all $X_k$ substrings are turned into $X_k^d$, except $X_i$ and $X_j$.

Finally, define a new additional character $\L$, which will be used to separate some of the components of our string.  We can now define $S$ and $T$.  We have
\begin{align*}
S &= \Z_{2p} \X \L = B_{2p}B_{2p-1} \ldots B_2B_1B_0 X_1 X_2 \ldots X_n \L 
\end{align*}

It follows from the definitions of $\Z_{2p}, \X$ and $\Delta$ that $S$ is exemplar.
Now for $i \in [p]$, define 

$$E_i :=  \Z_i^{01}\X_{e_i} \L \Z_{2p}^1\X \L$$ 

which we will call the \emph{edge gadget}.
Define $T$ as 
\begin{align*}
T &= \Z_{2p}^0 \X^d \L \Z_{2p}^1 \X \L E_1 E_2 \ldots E_p \\ &=\Z_{2p}^0 \X^d \L \Z_{2p}^1 \X \L 
\left[ \Z_1^{01}\X_{e_1} \L \Z_{2p}^1 \X \L \right]
\left[ \Z_2^{01}\X_{e_2} \L \Z_{2p}^1 \X \L \right]
%\left[ \Z_3^{01}\X_{e_3} \L \Z_{2p}^1 \X \L \right]
\ldots 
\left[ \Z_p^{01}\X_{e_p} \L \Z_{2p}^1 \X \L \right]
\end{align*}
(we add brackets for clarity --- they are not actual characters of $T$).
The idea is that $T$ starts with $S' = \Z_{2p}^0 \X^d \L$, a modified $S$ in which   $\Z_{2p}$ becomes $\Z_{2p}^0$ and the $X_i$ substrings are turned into $X_i^d$.  This $\X^d$ substring serves as a choice of vertices in our cost-effective subgraph.  Each edge $e_i$ has a ``gadget substring''
$E_i = \Z_i^{01}\X_{e_i} \L \Z_{2p}^1\X \L$.  Since $p$ is a multiple of $m$, the sequence of edge gadgets $E_1E_2 \ldots E_m$ is repeated $p/m$ times. 
Our goal to go from $T$ to $S$ is to get rid of all these edge gadgets by contractions.
Note that because a  $E_i$ gadget starts with $\Z_i^{01}$ and the gadget $E_{i+1}$ starts with $\Z_{i+1}^{01}$, the substring $E_{i+1}$ has a character that the substring  $E_i$ does not have.

\vspace{2mm}
\noindent
\textbf{The hardness proof.}
We now show that $G$ has a subgraph $W$ of cost at most $r$ if and only if $T$ can be contracted to $S$ using at most 
$p/m \cdot d (r + nm) + 4cdn$ moves.  We include the forward direction, which is the most instructive, in the main text.  The other direction can be found in the Appendix.
Although we shall not dig into details here, it can be deduced from the $(\Rightarrow$) direction that $T \contractstok{*} S$ holds.

($\Rightarrow$) Suppose that $G$ has a subgraph $W$ of cost at most $r$.  Thus $c(m - |E(W)|) + |W||E(W)| \leq r$.  To go from $T$ to $S$, first consider an edge $e_i$ that does not have both endpoints in $W$.  We show how to get rid of the gadget substring $E_i$ for $e_i$ using $dn + dc$ contractions.   
Note that $T$ contains the substring 
$\Z_{2p}^1\X \L E_i = \Z_{2p}^1\X \L [\Z_i^{01}\X_{e_i} \L \Z_{2p}^1\X \L]$, where brackets surround the $E_i$ occurrence that we want to remove.  
We can first contract $\Z_i^{01}$ to $\Z_i^{1}$ using $dc + 2d - 2$ contractions, then contract $\X_{e_i}$ to $\X$ using $d(n - 2)$
contractions.  The result is the $\Z_{2p}^1\X \L [\Z_i^1 \X \L \Z_{2p}^1\X \L]$ substring, which becomes $\Z_{2p}^1 \X \L$ using two contractions (see below).  This sums to $dc + 2d - 2 + dn - 2d + 2 = dc + dn$ moves.
More visually, the sequence of contractions works as follows (as usual brackets indicate the $E_i$ substring and what remains of it)
\begin{align*}
    &  \Z_{2p}^1\X \L \left[ \Z_i^{01}\X_{e_i} \L \Z_{2p}^1\X \L \right]  & \\
    \contractsto &\Z_{2p}^1\X \L \left[\Z_i^{1}\X_{e_i} \L\Z_{2p}^1\X \L \right]  & (dc + 2d - 2 \mbox{ contractions} ) \\
    \contractsto  &\Z_{2p}^1\X \L \left[ \Z_i^{1}\X \L\Z_{2p}^1\X \L\right]    & (d(n - 2) \mbox{ contractions} ) \\
    =&B_{2p}B_{{2p}-1} \ldots B_{i+1} 
     \Z_i^1 \X \L  \left[ \Z_i^1 \X \L  \Z_{2p}^1 \X \L \right] \\
    \contractsto & B_{2p}B_{{2p}-1} \ldots B_{i+1} 
    \Z_i^1 \X \L   \left[ \Z_{2p}^1 \X \L \right]& (1 \mbox{ contraction} )\\
    = & \Z_{2p}^1 \X \L \left[\Z_{2p}^1 \X \L \right] \\
    \contractsto & \Z_{2p}^1 \X \L & (1 \mbox{ contraction} )
\end{align*}
This sequence of $dn + dc$ contractions effectively removes the $E_i$ substring gadget.
Observe that after applying this sequence, it is still true that every remaining $E_j$ gadget substring is preceded by $\Z_{2p}^1\X\L$.
We may therefore repeatedly apply this contraction sequence to every $e_i$ not contained in $W$ (including those $e_i$ gadgets for which $i > m$).  This procedure is thus applied to $p/m \cdot (m - |E(W)|)$ gadgets.
We assume that we have done so, and that every $e_i$ for which the $E_i$ gadget substring remains is in $W$.  Call the resulting string $T'$.

Now, let $\X_W$ be the substring obtained from $\X^d$
by contracting, for each $v_i \in W$, the string $X_i^d$
to $X_i$.  We assume that we have contracted the $\X^d$ substring of $T'$ to $\X_W$, which uses $d|W|$ contractions (note that there is only one occurrence of $\X^d$ in $T'$, namely right before the first $\L$).  Call $T''$ the resulting string.  At this point, for every $E_i$ substring gadget that remains, where $E_i$ corresponds to edge $e_i = v_jv_k$, $\X_W$ contains the substrings $X_j$ and $X_k$ (instead of $X_j^d$ and $X_k^d$).

Let $i$ be the smallest integer for which the $e_i$ substring gadget $E_i$ is still in $T'$.
This is the leftmost edge gadget still in $T''$, meaning that $T''$ has the prefix
\begin{align*}
    \Z_{2p}^0\X_W \L \Z_{2p}^1 \X \L \left[ \Z_i^{01} \X_{e_i} \L \Z_{2p}^1 \X \L \right]
\end{align*}
where brackets indicate the $E_i$ substring.
To remove $E_i$, first contract $\Z_i^{01}$ to $\Z_i^{0}$, and contract $\X_{e_i}$ to $X_W$ (this is possible since $e_i \subseteq W$).  
The result is 
$\Z_{2p}^0 \X_W \L \Z_{2p}^1 \X \L \left[ \Z_i^0 \X_W \L \Z_{2p}^1 \X \L \right]$.  One more contraction gets rid of the second half.  This requires $dn + 2d - 1 + d(|W| - 2) + 1 = dn + d|W|$ contractions.
This procedure is applied to $p/m \cdot |E(W)|$ gadgets.  To recap, the contraction sequence for $E_i$ does as follows:
\begin{align*}
    &  \Z_{2p}^0\X_W \L \Z_{2p}^1\X\L \left[ \Z_i^{01}\X_{e_i} \L \Z_{2p}^1\X \L \right]  & \\
    \contractsto &\Z_{2p}^0\X_W \L \Z_{2p}^1\X\L \left[\Z_i^{0}\X_{e_i} \L\Z_{2p}^1\X \L \right] & (dn + 2d - 1 \mbox{ contractions} ) \\
    \contractsto  &\Z_{2p}^0\X_W \L \Z_{2p}^1\X\L \left[ \Z_i^{0}\X_W \L\Z_{2p}^1\X \L\right]    & (d(|W| - 2) \mbox{ contractions} ) \\
    \contractsto & \Z_{2p}^0 \X_W \L \Z_{2p}^1 \X \L & (1 \mbox{ contraction} )
\end{align*}
After we repeat this for every $E_i$, all that remains is the string $\Z_{2p}^0 \X_W \L \Z_{2p}^1 \X \L$.  We contract $\X_W$ to $\X$ using $d(n-|W|)$ contractions (in total, going from $\X^d$ to $\X$ required $dn$ moves).  Then contract $\Z_{2p}^0$ and $\Z_{2p}^1$ to $\Z_{2p}$ using $dc + 2d - 2 + dn + 2d - 1 = d(c + n + 4) - 3$ contractions.  One more contraction of the second half of the string yields $S$.
The summary of the number of contractions made is  
\begin{align*}
& \frac{p}{m} \cdot (m - |E(W)|) \cdot (dc + dn) + \frac{p}{m} \cdot |E(W)| \cdot (dn + d|W|) + dn + d(c + n + 4) - 3\\
\leq & \frac{p}{m} \cdot (m - |E(W)|) \cdot (dc + dn) + \frac{p}{m} \cdot |E(W)| \cdot (dn + d|W|) + 4cdn\\
= & \frac{p}{m} \cdot d \cdot (c + n)(m - |E(W)|) + \frac{p}{m} \cdot d \cdot (n + |W|)|E(W)| + 4cdn \\
=& \frac{p}{m} \cdot d \cdot \left[ c(m - |E(W)|) + |W||E(W)| + nm  \right]+ 4cdn \\
\leq& \frac{p}{m} \cdot d (r + nm)+ 4cdn
\end{align*}
as desired.

($\Leftarrow$): this direction of the proof is somewhat involved and we redirect the interested reader to the Appendix.  The idea is to show that a minimum contraction sequence must have the form similar to that in the ($\Rightarrow$) direction.  The challenging part is to show that each $E_i$ substring must get removed separately in this sequence, and that ``most'' of them incur a cost of either $dn + dt - 2$ or $dn + dc - 2$ for some $t$ (this ``most'' is the reason that we need a large $p$).

\section{An FPT algorithm for the exemplar problem}

In this section, we will show that \textsf{Exemplar-$k$-TD} can be solved in time $2^{O(k^2)} + poly(n)$ by obtaining a kernel of size $O(k2^k)$ (here $n$ is the length of $T$).

We first note that there is a very simple, brute-force algorithm to solve \textsf{$k$-TD}
(including \textsf{Exemplar-$k$-TD} as a particular case).
This only establishes membership in the $XP$ class, but 
it will be useful to evaluate the complexity of our kernelization later on.

\begin{proposition}\label{prop:k-search-tree}
The \textsf{$k$-TD} problem can be solved in time $O(n^{2k})$, where 
$n$ is the size of the target string.
\end{proposition}

\begin{proof}
Let $(S, T)$ be a given instance of \textsf{$k$-TD}.  Consider the branching algorithm that, starting from $T$, tries to contract every substring of the form $XX$ in $T$ and recurses on each resulting substring, decrementing $k$ by $1$ each time (the branching stops when $S$ is obtained or when $k$ reaches $0$ without attaining $S$). 
We obtain a search tree of depth at most $k$ and degree at most $n^2$, and thus it has $O(n^{2k})$ nodes.  Visiting the internal nodes of this search tree only requires enumerating $O(n^2)$ substrings, which form the set of children of the node.  Hence, there is no added computation cost to consider when visiting a node.   
%\qed
\end{proof}

From now on, we assume that we have an \textsf{Exemplar-$k$-TD} instance $(S, T)$, and so that $S$ is exemplar.

Let $x$ and $y$ be two consecutive characters in $S$ (i.e. $xy$ is a subtring of $S$).
We say that $xy$ is \emph{$(S,T)$-stable} if in $T$, every occurrence of $x$ in $T$ is followed by $y$
and every occurrence of $y$ is preceded by $x$.
An \emph{$(S,T)$-stable substring} $X = x_1 \ldots x_l$, where $l \geq 2$, is a substring 
of $S$ such that $x_ix_{i+1}$ is $(S,T)$-stable for every $i \in [l - 1]$.  We also define a string with a single character $x_i$ to be a $(S, T)$-stable substring (provided $x_i$ appears in $S$ and $T$).
If any substring of $S$ that strictly contains $X$ is not an $(S,T)$-stable substring, then 
$X$ is called a \emph{maximal $(S,T)$-stable substring}.
Note that these definitions are independent of $S$ and $T$, and so the same definitions apply for $(X, Y)$-stability, for any strings $X$ and $Y$.

We will show that every maximal $(S, T)$-stable substring can be replaced by a single character, 
and that if $T$ can be obtained from $S$ using at most $k$ tandem duplications, 
then this leaves strings of bounded size.

We first show that, roughly speaking, stability is maintained by all tandem duplications 
when going from $S$ to $T$.

\begin{lemma}\label{lem:stay-stable}
Suppose that $dist_{TD}(S, T) = k$ and let $X$ be an $(S,T)$-stable substring.
Let $S = S_0, S_1, \ldots, S_k = T$ be any minimum sequence of strings transforming $S$ to $T$ by tandem duplications. 
Then $X$ is $(S, S_i)$-stable for every $i \in [k]$.
\end{lemma}

\begin{proof}
Assume the lemma is false, and let $S_i$ be the first of $S_1, \ldots, S_k$ that does not verify the statement.
Then there are two characters $x, y$ belonging to $X$ such that $xy$ is $(S, T)$-stable, 
but $xy$ is not $(S, S_i)$-stable.

We claim that, under our assumption, $xy$ is not $(S, S_j)$-stable for any $j \in \{i, \ldots, k\}$.
As this includes $S_k = T$, this will contradict that $xy$ is $(S, T)$-stable.
We do this by induction --- as a base case, $xy$ is not $(S, S_i)$-stable so this is true for $j = i$.  
Assume that $xy$ is not $(S, S_{j-1})$-stable, where $i < j \leq k$.
Let $D$ be the duplication transforming $S_{j-1}$ to $S_j$ (here $D = (a,b)$ contains the start and 
end positions of the substring of $S_{j-1}$ to duplicate).

Suppose first that $xy$ is not $(S, S_{j-1})$-stable because $S_{j-1}$ has an occurrence of $x$ 
that is not followed by $y$.  
Thus $S_{j - 1}$ has an occurrence of $x$, say at position $p_x$, followed by $z \neq y$.  If we assume that $xy$ is $(S, S_j)$-stable, then 
a $y$ character must have appeared after this $x$ from $S_{j-1}$ to $S_j$. 
Changing the character next to this $x$ is only possible if the last character duplicated by $D$ is 
the $x$ at position $p_x$ and the first character of $D$ is a $y$.  
In other words, denoting $S_{j-1} = A_1yA_2xzA_3$ for appropriate $A_1, A_2, A_3$ substrings, the $D$ duplication must do the following
\begin{align*}
A_1\underline{yA_2x}zA_3 \dupsto A_1\underline{yA_2xyA_2x}zA_3
\end{align*}

But then, there is still an occurrence of $x$ followed by $z$, and it follows that $xy$ cannot be $(S, S_j)$-stable.

So suppose instead that $xy$ is not $(S, S_{j-1})$-stable because $S_{j-1}$ has an occurrence of $y$ preceded by $z \neq x$.
The character preceding this $y$ has changed in $S_j$.  But one can verify that this is impossible.  For completeness, we present each possible case: either $D$ includes both $z$ and $y$, includes one of them 
or none.  These cases are represented below, and each one of them leads to an occurrence of $y$ still preceded by $z$
(the left-hand side represents $S_{j-1}$ and the right-hand side represents $S_j$):
\begin{align*}
\mbox{Include both: } &A_1\underline{A_2zyA_3}A_4 \Rightarrow A_1\underline{A_2zyA_3A_2zyA_3}A_4 \\
\mbox{Include $z$ only: } &A_1\underline{A_2z}yA_3 \Rightarrow A_1\underline{A_2zA_2z}yA_3 \\
\mbox{Include $y$ only: } &A_1z\underline{yA_2}A_3 \Rightarrow A_1z\underline{yA_2yA_2}A_3 \\
\mbox{Include none: } &A_1\underline{A_2}zyA_3 \Rightarrow A_1\underline{A_2A_2}zyA_3 
\mbox{ or } A_1zy\underline{A_2}A_3 \Rightarrow A_1zy\underline{A_2A_2}A_3
\end{align*}

We have therefore shown that $xy$ cannot be $(S, S_j)$-stable, and therefore not $(S, T)$-stable, which conludes the proof.  
%\qed
\end{proof}

Let $S'$ be a substring obtained from $S$ by tandem duplications, and let $X := S'[a .. b]$ be the substring of $S'$ at positions from $a$ to $b$.  Suppose that we apply a duplication $D = (c, d)$, which copies the substring $S'[c .. d]$. 
Then we say that $D$ \emph{cuts} $X$ if 
one of the following holds:

\begin{itemize}
    \item 
    $a < c \leq b$ and $b < d$, in which case we say that $D$ cuts $X$ \emph{to the right};
    
    \item
    $c < a$ and $a \leq d < b$, in which case we say that $D$ cuts $X$ \emph{to the left};
    
    \item
    $(a, b) \neq (c, d)$ and $a \leq c < d \leq b$, in which case $D$ cuts $X$ \emph{inside}.
\end{itemize}

In other words, if we write $X = X_1X_2$ and $S' = UV X_1X_2 WY$, cutting to the right takes the form $UV X_1 \underline{X_2W} Y \Rightarrow UB X_1 \underline{X_2WX_2W} Y$.  Cutting to the left takes the form $U \underline{VX_1}X_2 WY \Rightarrow U \underline{VX_1VX_1} X_2 WY$.  Rewriting $S' = U X_1X_2X_3 V$, cutting inside takes the form $U X_1 \underline{X_2} X_3 V \Rightarrow U X_1 \underline{X_2X_2} X_3 V$.
Note that if $D$ does not cut any occurrence of a maximal $(S, S')$-stable substring $X$ and $S''$ is obtained by applying $D$ on $S'$, then $X$ is $(S, S'')$-stable. 

The next lemma shows that we can assume that maximal stable substrings never get cut, and thus always get duplicated together.  The proof is in the Appendix: the idea is that any duplication that cuts an $X_j$ can be replaced by an equivalent duplication that doesn't.

\begin{lemma}\label{lem:never-cut-maximals}
Suppose that $dist_{TD}(S, T) = k$, and let $X_1, \ldots, X_l$ be the set of maximal $(S, T)$-stable substrings.  Then there exists a sequence of 
tandem duplications $D_1, \ldots, D_k$ transforming $S$ into $T$ such that no occurrence of an $X_j$ gets cut by a $D_i$. 

In other words, for all $i \in [k]$ and all $j \in [l]$, $D_i$ does not cut any occurrence of $X_j$ in the string obtained by applying $D_1, \ldots, D_{i-1}$ to $S$.
\end{lemma}

\begin{proof}
Let $D^*_1, \ldots, D^*_k$ be a sequence of tandem duplications transforming $S$ into $T$, 
and for $i \in [k]$, let $S_i$ be the string obtained by applying the first $i$ duplications.
Put $S_0 := S$.
We show that any $S_i$, $i \in [k]$, can be obtained from $S_{i-1}$ by a duplication $D_i$ that does not cut 
any $X_j$ occurrence in $S_{i-1}$, $j \in [l]$.  This proves the lemma, since $D_1, \ldots, D_k$ will form the desired sequence of duplications.

Fix $i \in [k]$, and assume that $D^*_i$ cuts some of the $X_j$'s.  
We note that since $S$ is exemplar, the $X_j$'s have pairwise distinct characters.  Hence $D^*_i$ can cut at most 
two occurrences of a maximal $(S, S_i)$-stable substrings, at most one to the left and at most one to the right (if an $X_j$ is cut inside, only one string can get cut).
Also, by Lemma~\ref{lem:stay-stable}, we know that every $X_j$ substring is $(S, S_i)$-stable.
We have four cases to consider: 

\begin{itemize}

\item 
$D^*_i$ cuts some $X_j$ inside.  Write $X_j = X_j^1X_j^2X_j^3$, where at least one of $X_j^1$ or $X_j^2$ is non-empty, and $S_{i-1} = AX_j^1X_j^2X_j^3B$.  This results in 

\begin{align*}
    S_{i-1} = AX_j^1\underline{X_j^2}X_j^3B \Rightarrow AX_j^1\underline{X_j^2X_j^2}X_j^3B = S_{i}
\end{align*}

Since characters from $X_j^1, X_j^2$ and $X_j^3$ are pairwise disjoint, $X_j$ cannot be $(S, S_{i})$-stable, a contradiction of Lemma~\ref{lem:stay-stable}.

\item
$D^*_i$ cuts some $X_j$ to the right, but no other string to the left.
%Then there are substrings $A,B,C, X_j^1, X_j^2$ such that
Then we may write $X_j$ and $S_{i-1}$, respectively, as
 $X_j = X_j^1 X_j^2$ and $S_{i-1} = AX_j^1 X_j^2 BC$ such that $D$ copies the substring $X_j^2B$.  This gives

\begin{align*}
S_{i-1} = AX_j^1 \underline{X_j^2 B} C \Rightarrow AX_j^1 \underline{X_j^2 B X_j^2 B} C = S_{i}
\end{align*}

But by Lemma~\ref{lem:stay-stable}, $X_j$ is $(S, S_{i})$-stable.  Since $S_{i}$ has $BX_j^2$  as a substring, this must mean that $B = \hat{B} X_j^1$ for some substring 
$\hat{B}$ (note that we use the fact that $X_j^2$ has distinct characters, and thus that the occurrence of $X_j^1$ must be entirely in $B$).  Therefore $S_{i-1} = AX_j^1 X_j^2 \hat{B} X_j^1 C$.  Since $X_j$ is also $(S, S_{i-1})$-stable, this in turn implies that $C = X_j^2 \hat{C}$ for some substring $\hat{C}$, and in fact we get 

\begin{align*}
S_{i-1} = AX_j^1 \underline{X_j^2 \hat{B} X_j^1} X_j^2 \hat{C} \Rightarrow AX_j^1 \underline{X_j^2 \hat{B} X_j^1} \underline{X_j^2 \hat{B} X_j^1} X_j^2 \hat{C} = S_{i}
\end{align*}

We can replace $D^*_i$ by a duplication that copies $X_j^1X_j^2\hat{B}$, i.e.

\begin{align*}
S_{i-1} = A \underline{X_j^1 X_j^2 \hat{B} }X_j^1 X_j^2 \hat{C} \Rightarrow A \underline{X_j^1 X_j^2 \hat{B} X_j^1 X_j^2 \hat{B} }X_j^1 X_j^2 \hat{C} = S_{i}
\end{align*}

Since this duplication starts with $X_j$ and copies itself right before another occurrence of $X_j$, it is clear that it does not cut any maximal $(S, T)$-stable substring, as desired.

\item
$D^*_i$ cuts some $X_j$ to the left, but cuts no string to the right.
Then we may write 

\begin{align*}
S_{i-1} = A\underline{BX_j^1}X_j^2C \Rightarrow A\underline{B X_j^1 B X_j^1}X_j^2C = S_{i}
\end{align*}

Similarly as in the previous case, since $X_j$ is $(S, S_{i})$-stable, we must have $B = X_j^2\hat{B}$.  We are led to deduce that 
$A = \hat{A}X_j^1$.
Therefore we have

\begin{align*}
S_{i-1} = \hat{A}X_j^1\underline{X_j^2\hat{B} X_j^1}X_j^2 C \Rightarrow 
\hat{A}X_j^1\underline{X_j^2\hat{B} X_j^1 X_j^2\hat{B} X_j^1}X_j^2 C = S_{i}
\end{align*}

As before, we could instead duplicate the substring $X_j^1X_j^2 \hat{B}$ occuring right after $\hat{A}$.

\item
$D^*_i$ cuts some $X_j$ to the left and some $X_h$ to the right.  Note that $X_j = X_h$ is possible, which we will in fact show to hold.
%To make the notation lighter, we will write $Y := X_j$ and $Z := X_h$.
We may write $X_j = X_j^1X_j^2$ and $X_h = X_h^1 X_h^2$ such that we get

\begin{align*}
S_{i-1} = AX_j^1\underline{X_j^2 B X_h^1}X_h^2 C \Rightarrow AX_j^1\underline{X_j^2 B X_h^1 X_j^2 B X_h^1} X_h^2 C = S_{i}
\end{align*}

Now, $X_j$ is $(S, S_{i})$-stable and $S_{i}$ contains $X_h^1 X_j^2$ as a substring.
It follows that the last character of $X_h^1$ must be the last character of $X_j^1$ (applying the $(S, X_j)$-stability argument on the $X_h^1X_j^2$ substring).  
In other words, $X_h$ and $X_j$ have a character in common.
Since $S$ is exemplar, the set of maximal $(S, T)$-stable strings $X_1, \ldots, X_l$
have pairwise disjoint sets of characters and partition $S$ into substrings.  We deduce that $X_j = X_h$, as we predicted. 

We now want to show that $X_h^1 = X_j^1$.  Note that both $X_j^1$ and $X_h^1$ are prefixes of $X_j$
(for $X_h^1$, this is because $X_h = X_j$).
Moreover, as argued the last character of $X_h^1$ is also the last character of $X_j^1$.
These two observations establish that $X_h^1 = X_j^1$ (and therefore $X_h^2 = X_j^2$).
This allows us to rewrite $S_i$ and $S_{i+1}$ as 

\begin{align*}
S_{i-1} = AX_j^1 \underline{X_j^2 B X_j^1}X_j^2 C \Rightarrow AX_j^1\underline{X_j^2 B X_j^1 X_j^2 B X_j^1} X_j^2 C = S_{i}
\end{align*}

It becomes clear that we can duplicate the $X_j^1X_j^2B$ substring after $A$ in $S_{i-1}$ to obtain 
$S_{i}$.  This duplication does not cut any maximal $(S, T)$-stable substring. 
%\qed
\end{itemize}

We have thus shown that if $D_i^*$ cuts some occurrence of one or more of the $X_j$'s, then $D_i^*$ can be replaced by a duplication $D_i$ that yields the same string $S_i$ as $D_i^*$.  
The only case remaining is when $D_i^*$ does not cut any $X_j$.  In that case, we set $D_i = D_i^*$.  This shows that we can find the claimed sequence $D_1, \ldots, D_k$ in which no $X_j$ ever gets cut.
\end{proof}

The above implies that we may replace each maximal $(S, T)$-stable substring $X$
of $S$ and $T$ by a single character, since we may assume that characters of $X$ 
are always duplicated together.
It only remains to show that the resulting strings are small enough.
The proof of the following lemma has a very simple intuition.  First, $S$ has exactly $1$ maximal $(S, S)$-stable substring.  Each time we apply a duplcation, we ``break'' at most 2 stable substrings, which creates 2 new ones.  So if we apply $k$ duplications, there are at most $2k + 1$ such substrings in the end.
See the Appendix for the full proof.

\begin{lemma}\label{lem:bounded-maximals}
If $dist_{TD}(S, T) \leq k$, then there are at most $2k + 1$ maximal $(S, T)$-stable substrings.
\end{lemma}

\begin{proof}
Let $S = S_0, S_1, \ldots, S_k = T$ be any minimum sequence of strings transforming $S$ to $T$ by tandem duplications.  
We show by induction that, for each $i \in \{0, 1, \ldots, k\}$, the number of maximal $(S, S_i)$-stable substrings is at most $2i + 1$.  
For $i = 0$, there is only one maximal $(S, S)$-stable substring, namely $S$ itself. 
Now assume that there are at most $2(i - 1) + 1 = 2i - 1$ maximal $(S, S_{i-1})$-stable substrings.
Let $\X = \{X_1, \ldots, X_l\}$ be the set of these substrings, $l \leq 2i - 1$.  We then know that $S_{i-1}$ can be written as a concatenation of $X_j$'s from $\X$ (with possible repetitions).  The duplication $D$ transforming $S_{i-1}$ to $S_i$ copies some of these $X_j$'s entirely, except at most two $X_j$'s at the ends which it may copy partially (i.e. $D$ cuts at most two substrings from $\X$).  In other words, the substring duplicated by $D$ can be written as $X_j^2 X_{a_1} X_{a_2} \ldots X_{a_r} X_h^1$, where $X_j = X_j^1X_j^2$ and $X_h = X_h^1X_h^2$ for some $j, h \in [l]$ (and $X_{a_1}, \ldots, X_{a_r} \in \X)$.
Going further, $S_{i-1}$ and $S_i$ can be written, using appropriate substrings $A, B, C$ that are concatenation of elements of $\X$, as

\begin{align*}
    S_{i-1} = A X_j^1 \underline{ X_j^2 B X_h^1 }  X_h^2 C \Rightarrow
    A X_j^1 \underline{X_j^2 B X_h^1 X_j^2 B X_h^1 } X_h^2 C = S_i
\end{align*}

Now, any $X_r \in \X \setminus \{X_j, X_h\}$ is $(S, S_i)$-stable.  
Moreover, $X_j^1, X_j^2, X_h^1$ and $X_h^2$ are also $(S, S_i)$-stable.  
This shows that the number of maximal $(S, S_i)$-stable substrings 
is at most $2i - 1 - 2 + 4 = 2i + 1$, as desired.
%\qed
\end{proof}

We can now transform an instance $(S, T)$ of \textsf{Exemplar-$k$-TD} to a kernel, an equivalent instance $(S', T')$ of size depending only on $k$.

\begin{theorem}
An instance $(S, T)$ of \textsf{Exemplar-$k$-TD}  admits a kernel $(S', T')$ in which $|S'| \leq 2k + 1$ and $|T'| \leq (2k + 1)2^k$. 
\end{theorem}

\begin{proof}
Let $S', T'$ be obtained from an instance $(S, T)$ by replacing each maximal 
$(S, T)$-stable substring by a distinct character.  
We first prove that $(S', T')$ is indeed a kernel by establishing its equivalence with $(S, T)$.
Clearly if $(S', T')$ can be solved using at most $k$ duplications, then the same applies to $(S, T)$.
By Lemma~\ref{lem:never-cut-maximals}, the converse also holds: if $(S, T)$ can be solved with at most $k$ duplications, we may assume that these duplications never cut a maximal $(S, T)$-stable substring, and so these duplications can be applied on $(S', T')$.  

Then by Lemma~\ref{lem:bounded-maximals}, we know that $S'$ has at most $2k + 1$
characters.  If $dist_{TD}(S', T') \leq k$, then each duplication can at most double the size of the previous string.  Therefore, $T'$ must have size at most $(2k + 1)2^k$.
%\qed
\end{proof}

The kernelization can be performed in polynomial time, as one only needs to identify maximal $(S, T)$-stable substrings and contract them (we do not bother with the exact complexity for now).  
Running the brute-force algorithm from Proposition~\ref{prop:k-search-tree}
yields the following.

\begin{corollary}
The exemplar $k$-tandem duplication problem can be solved in time 
$O(((2k + 1)2^k)^{2k} + poly(n)) = 2^{O(k^2)} + poly(n)$, where $n$ is the size of the input.
\end{corollary}

%It is an open question whether the \textsf{Exemplar-$k$-TD}, or even \textsf{$k$-TD} for that matter, admit a polynomial kernel.  Noting that the exemplar condition was used several times in our proofs, the FPT status of \textsf{$k$-TD} also remains open.

\section{Open problems}

Although this work answers some open questions, many of them still deserve investigation.  We conclude with some of these question along with future research  perspectives.

\begin{itemize}
    \item 
    Is the \textsf{$k$-TD} problem FPT in parameter $k$?  As we observe in our \textsf{Exemplar-$k$-TD} kernelization, if $T$ and $S$ are large compared to $k$, they must share many long common substrings which could be exploited for an FPT algorithm.  It is also an interesting question whether \textsf{Exemplar-$k$-TD} admits a polynomial size kernel.

    \item
    If $|\Sigma|$ is fixed, is \textsf{$k$-TD} in $P$?  Even the $|\Sigma| = 2$ case is open.  One possibility it to check whether we can reduce the alphabet of any instance to some constant by encoding each character appropriately. 
    
    \item
    Can one decide in polynomial time whether $S \dupstok{*} T$?  
    The only known result on this topic is that it can be done if $|\Sigma| = 2$, 
    as one can construct a finite automaton accepting all strings generated by $S$
    (though this automaton does not give the minimum number of duplications required).

    \item
    Does the \textsf{$k$-TD} problem admit a constant factor approximation algorithm?  The answer might depend on the hardness of deciding whether $S \dupstok{*} T$, but one might still consider the promise version of the problem.

    \item
    If the length of each duplicated string is bounded by $d$, is \textsf{$k$-TD} in $P$ (with $d$ treated as a constant)?  We believe that it is FPT in $k + d$, but is it FPT in $d$?
    
\end{itemize}

%\bibliographystyle{plainurl}

%\bibliography{main}

\newpage
\section*{Appendix}

\subsection*{Proof of Theorem~\ref{thm:main-np-hardness}, ($\Leftarrow$) direction}

\noindent
%($\Leftarrow$) 
Suppose that $T$ can be turned into $S$ using $\alpha$ contractions, where $\alpha \leq p/m \cdot d (r + nm) + 4cdn$.  Let $C_1, \ldots, C_{\alpha}$ be a corresponding sequence of contractions.  Here, each $C_i$ contraction is given by a pair of positions ranging over both copies of the contracted substring.
The idea is to show that, for some integer $t$, many of the $E_i$ substrings are removed after $t$ of the $X_i^d$ substrings from $\X^d$ have been contracted to $X_i$.  This set of $t$ $X_i$'s corresponds to the vertices of a cost-effective subgraph.  The main components of the proof are to show that each $E_i$ must be removed, no two $E_i$'s are affected by the same contraction, and most (though perhaps not all) $E_i$ require either $dn + dt$ or  $dc + dn - 1$ contractions.

Denote $T(l)$ as the string obtained from $T$ after applying the first $l$ contractions $C_1, \ldots, C_l$ in the sequence, with $T(0) = T$ and $T(\alpha) = S$.
A \emph{block} of $T(l)$  is a substring $P$ of $T(l)$ whose last character is $\L$, that has only one occurrence of $\L$ and that is a maximal string with this property (hence in $T(l)$, the first character of $P$ is either preceded by $\L$ or is the start of $T(l)$).  For instance, each $E_i$ substring is made of 2 blocks.

We need a (conceptual) mapping from the characters of $T(l)$ to those of $T$.  We assume that each character of $T$ is distinguishable, i.e. each character has a unique identifier associated to it (we do not define it explicitly, but for instance each character can be labeled by its position in $T$) .  When contracting a substring $DD$ from $T(l)$ to $T(l + 1)$, we assume that the characters of the second half are deleted.  
That is, if $T(l) = LDDR$ and $T(l + 1) = LDR$, only the characters from the first, leftmost $D$ substring remain.  Therefore when going from $T(l)$ to $T(l + 1)$, some characters might change position but they keep the same identifier.  Thus each character of $T(l)$ corresponds to a distinct character in $T$, namely the one with the same identifier.  When we say that a character $x$ from $T(l)$ \emph{belongs} to a subtring $P$ of $T$, we mean that $x$ corresponds to a character of $P$ in $T$ under this mapping.

For a substring $P$ of $T$, we say that $P$ is \emph{removed} in $T(l)$ if $T(l)$ has no characters that belong to $P$. 
We say that $P$ is removed if there is some $T(l)$ in which $P$ is removed.

\begin{nclaim}
Each $E_i$ substring must be removed in $T(\alpha)$.
\end{nclaim}

\begin{proof}
Consider the first, leftmost block $\Z_{2p}^0 \X^d \L$ of $T$.  
Observe that for any $T(l)$ and any symbol $s \in \Sigma$, $T(l)$ has an occurrence of $s$ that belongs to this block (as there is no way to completely remove all occurrences of a symbol from the first block $\Z_{2p}^0 \X^d \L$, by our way of deleting the rightmost copy in contractions). 
Since $\Sigma(E_i) \subseteq \Sigma(\Z_{2p}^0 \X^d \L)$, this means that if $E_i$ is not removed, the last string $T(\alpha)$ in the sequence has at least two occurrences of some character in $\Sigma$.  Because $S$ is exemplar, this contradicts that $T(\alpha) = S$.
%\qed
\end{proof}

Notice that in $T$, $E_i = \Z_i^{01} \X_{e_i} \L \Z_{2p}^1 \X \L$ has two blocks.  We write $E_i' = \Z_i^{01} \X_{e_i} \L$ to denote the first block of $E_i$.
We let $E'_i(l)$ be the substring of $T(l)$ formed by all the characters that belong to $E'_i$, noting that $E'_i(l)$ is possibly the empty string or a subsequence of $E'_i$.  
For $a \in \{0,1,01\}$, a block $BX\L$ is called a $\Z^{a}_i\X\L$-block if $BX\L$ is a subsequence of $\Z_i^{a}\X^d\L$ and $\Sigma(BX\L) = \Sigma(\Z_i^{a}\X^d\L)$.  In other words, $BX\L$ has all the symbols that occur in $\Z_i^{a}\X^d\L$ in the same order, although the number of occurrences of a symbol might differ.
A $\Z_{2p}^1\X\L$-cluster is a string obtained by concatenating an arbitrary number of $\Z_{2p}^1\X\L$-blocks.  Using notation borrowed from regular languages, we write $(\Z_{2p}^1\X\L)^*$ to denote a possibly empty $\Z_{2p}^1\X\L$-cluster.

\begin{nclaim}\label{claim:the-form}
For any $l$, $T(l)$ has the form 

$$BX\L(\Z_{2p}^1\X\L)^* E'_{i_1}(l) (\Z_{2p}^1\X\L)^* E'_{i_2}(l) (\Z_{2p}^1\X\L)^* \ldots E'_{i_h}(l) (\Z_{2p}^1\X\L)^*$$
where

\begin{itemize}
    \item 
    $B X \L$ is a $\Z_{2p}^0\X\L$-block

    \item
    $i_1 < i_2 < \ldots < i_h$

    \item
    each $(\Z_{2p}^1\X\L)^*$ is a $\Z_{2p}^1\X\L$-cluster
    
    \item 
    for each $j \in \{i_1, \ldots, i_h\}$, $E'_j(l)$ is a $\Z_j^{01}\X\L$-block
    
\end{itemize}
\end{nclaim}

\begin{proof}
Notice that the statement is true for $l = 0$, since $T$ has the required form.  Assume the claim is false and let $l$ be the smallest integer for which $T(l)$ is a counter-example to the claim.  Thus we may assume that $T(l - 1)$ has the same form as in the claim statement.
Let $D$ be the string that was contracted from $T(l - 1)$ to $T(l)$ (so that $T(l - 1)$ contained $DD$ as a subtring, and the second $D$ substring gets removed from $T(l-1)$).  If $D$ does not contain a $\L$ character, then $DD$ is entirely contained in a single block.  Contracting $DD$ cannot remove all occurrences of a symbol nor change their order, and thus the above form must be preserved (every $\Z_i^a\X\L$-block will remain a $\Z_i^a\X\L$-block).
Assume instead that the last character of $D$ is $\L$.  Then $DD = D'\L D'\L$ for some string $D'$, and removing the second $D'\L$ half only removes entire blocks of $T(l - 1)$.  As this block cannot be $BX\L$ and since each $E'_{i_j}(l - 1)$ is itself a block, this preserves the form of the claim.

Therefore, we may assume that the last character of $D$ is not $\L$, but that $D$ has at least one $\L$ character.  Observe that no character from the $BX\L$-block can get removed by such a contraction, since the left half of $DD$ is kept.  It follows that the first condition of the claim is preserved after contracting $DD$.  It is easy to see that the second condition is also preserved.  
For the other two conditions, we have four cases to consider depending on where the right half of $DD$, i.e. the removed substring, is located in $T(l-1)$.

\begin{enumerate}
    \item 
    The leftmost character removed belongs to a $E'_{j}(l-1)$ substring.  In this case, because $D$ contains a $\L$, the right half of $DD$ must contain the $\L$ of $E'_{j}(l - 1)$.  Let $b$ be the first character of $E'_j(l-1)$, which is the first character of $\Z_j^{01}$ since $E'_j(l-1)$ is a $\Z_j^{01}\X\L$-block, by assumption.  We treat $b$ as a uniquely identifiable character in $T(l-1)$.  Note that this $b$ is preceded by $\L$ in $T(l - 1)$.
    %Write $E'_j(l-1) = b\hat{E_j}\L$ for convenience, and we know that a suffix of $b \hat{E_j} \L$ gets removed.  
    There are two subcases: either this $b$ is the leftmost removed character or not.  In the first case, $D = bD'$ for some $D'$, which we illustrate as follows (we add brackets around the two copies of $D$, and underline the removed half):
    
    $$T(l-1) = T' [bD'] [\underline{bD'}] T''$$
    
    for some $T'$ and $T''$.  Here the second $b$ is the one from $E'_j(l-1)$.  Since it is preceded by $\L$ in $T(l-1)$, this implies that $D'$ (and thus $D$) ends with a $\L$.  But we are assuming that $D$ does not end with $\L$.
    Therefore we know that $b$ is not the leftmost character removed from $T(l-1)$.  In this case, the $b$ belongs to the left half of $DD$ (if not, the left $D$ would entirely be in $E'_j(l-1)$ and could not contain a $\L$).  This case can be illustrated as follows:
    
    $$T(l-1) = T' [D'\L b D'']  [\underline{D' \L b D'' }] T''$$
    
    where $D = D' \L b D''$.  
    Here, the first $b$ is the one from $E'_j(l-1)$.  Why does the left $D$ have to contain the $\L$ preceding $b$?  Because we know $D$ contains $\L$: if the left $D$ starts with $b$ and contains $\L$, it contains all of $E'_j(l-1)$, contradicting that characters of $E'_j(l-1)$ get deleted.  
    It follows that $D$ must contain $\L b$ as a substring.  But  there is only one occurrence of $\L b$ in $T(l-1)$, as  $E'_{j}(l-1)$ is the only block that starts with $b$.  
    Therefore, $T(l-1)$ cannot contain $DD$ as a subtring, a contradiction.  
    
    \item
    The rightmost character removed is in some 
$E'_{j}(l-1)$ substring.  Again, if we put $b$ as the first character of $E'_j(l-1)$, this means that the removed $D$ contains $\L b$ as a subtring (if not, $D$ cannot contain a $\L$), which has only one occurrence.  We get the same contradiction.

    \item
    The leftmost and rightmost characters that get removed belong to distinct $(\Z_{2p}^1 \X \L)$-clusters, implying the existence of at least one $E'_j(l-1)$ in between.
    The same type of $\L b$ substring argument applies, since the removed $D$ contains the first character of $E'_j(l-1)$ and its preceding $\L$.
    
    \item
    The leftmost and rightmost characters that get removed belong to the same 
    $(\Z_{2p}^1 \X \L)$-cluster.  In this case, it is not hard to verify that the result is yet another $(\Z_{2p}^1 \X \L)$-cluster, which preserves the desired form.
\end{enumerate}

The cases above cover every possibility: we have covered the cases where the removed substring begins or ends in a $E'_j(l-1)$, and the cases where both its extremities end in a cluster.  This proves the claim.
\end{proof}

We will say that a contraction $C_l$ \emph{affects} $E'_i(l)$
if at least one character of $E'_i(l)$ is in the substring corresponding to $C_l$.  Recall that $C_l$ spans over both copies of the contracted substring, and so $E'_i(l)$ could be affected by $C_l$ even if none of its characters gets removed.

\begin{nclaim}\label{claim:no-two-eis}
For any $l$, the contraction $C_l$ from $T(l)$ to $T(l+1)$ does not affect two distinct $E'_i(l)$ and $E'_j(l)$ substrings of $T(l)$.
\end{nclaim}

\begin{proof}
Suppose the claim is false, and let $T(l)[a_1 .. a_2]$ be the substring of $T(l)$ affected by the contraction, where $T(l)[a_1 .. a_2] = DD$ for some string $D$.  Assume that $T(l)[a_1 .. a_2]$ contains characters from both $E'_i(l)$ and $E'_j(l)$, where $i < j$.  Let $b_i, b_j$ be the first characters of $E'_i(l)$ and $E'_j(l)$, respectively, which are the first character of $\Z_i^{01}$ and $\Z_j^{01}$ by Claim~\ref{claim:the-form}.
Then $T(l)[a_1 .. a_2]$ must contain the substring $\L b_j$, since $E'_j(l)$ occurs later than $E'_i(l)$ in $T(l)$.  Since $\L b_j$ occurs only once in $T(l)$ as argued in the previous claim, $\L b_j$ cannot be a substring of $D$.  This is only possible if  $D$ starts with $b_j$ (and consequently ends with $\L$).
Now, since $E'_i(l)$ does not contain $b_j$, $T(l)[a_1 .. a_2]$ cannot start with a suffix of $E'_i(l)$.  Yet some characters of $E'_i(l)$ are in $T[a_1 .. a_2]$, implying that the substring $\L b_i$ is in $T(l)$.  Again, this substring occurs only once in $T(l)$, and thus $D$ must start with $b_i$ and end with $\L$.  But this is impossible since $b_i \neq b_j$.
%\qed
\end{proof}

Notice that $T$ has one occurrence of the $\X^d = X_1^d \ldots X_n^d$ substring.  We will therefore refer to the $\X^d$ substring of $T$ without ambiguity.
For $i \in [n]$, we let 
$X_i(l)$ denote the substring of $T(l)$ formed by all the characters that belong to the $X_i^d$ substring of $\X^d$.
We will say that $X_i$
is \emph{activated} in $T(l)$ if $X_i(l) = X_i$.
Intuitively speaking, $X_i$ is activated in $T(l)$ if it has undergone $d$ contractions to turn it from $X_i^d$ into $X_i$.

\begin{nclaim}\label{claim:two-removal-ways}
Let $i \in [p]$, and suppose that $E_i'$ is not removed in $T(l- 1)$ but is removed in $T(l)$.  Let $t$ be the number of $X_i$'s that were activated in $T(l - 1)$.  Suppose that $v_{i_1}$ and $v_{i_2}$ are the two endpoints of edge $e_i$.

Then the number of contractions that have affected $E'_i$ is at least $dc + dn - 1$ if $X_{i_1}$ or $X_{i_2}$ is not activated in $T(l-1)$, or at least $\min\{ dt + dn , dc + dn - 1 \}$ if $X_{i_1}$ and $X_{i_2}$ are both activated in $T(l-1)$.
\end{nclaim}

\begin{proof}
By Claim~\ref{claim:the-form}, in $T(l-1)$, $E'_i(l - 1)$ belongs to a $\Z_i^{01}\X\L$-block.
As $E'_i(l - 1)$ gets removed completely after the $l$-th contraction of some substring $DD$, it follows that $D$ must contain a substring that is equal to $E'_i(l - 1)$.  %, and this substring must be from a block that contains $E'_i(l - 1)$ as a substring.
The second $D$ of the $DD$ square certainly contains the $E'_i(l - 1)$ substring that gets removed, but consider the copy of $E'_i(l-1)$ in the first $D$ of the $DD$ square.  That is, we can represent the contraction as

$$
T' [D_1 \hat{E}'_i(l-1) D_2] [\underline{D_1 E'_i(l-1) D_2}] T''
$$

where $D = D_1 E'_i(l-1) D_2$ and $\hat{E}'_i(l-1)$ is a substring equal to $E'_i(l-1)$.
Since $E'_i(l-1)$ is a block, this $\hat{E}'_i(l-1)$ copy is a substring of a (possibly larger) block.
By Claim~\ref{claim:no-two-eis}, there are only two such possible blocks: either it is $BX\L$, which is the $\Z_{2p}^0\X\L$-block at the start of $T(l-1)$, or it is a $\Z_{2p}^1\X\L$-block from a cluster preceding $E'_i(l - 1)$.  We analyze these two cases, which will prove the two cases of the claim.

Suppose that $\hat{E}'_i(l-1)$ is located in the first block $BX\L$ of $T(l-1)$.  Note that since $\X_{e_i}$ contains $X_{i_1}$ and $X_{i_2}$ in their contracted form (as opposed to $X^d_{i_1}$ or $X^d_{i_2}$), $X_{i_1}$ and $X_{i_2}$ must be activated in $T(l-1)$ for the $DD$ contraction to be possible.  
Moreover for $E'_i(l-1)$ to be equal to a substring of $BX\L$, every other $X_j$ with $j \neq i_1, i_2$ that is activated must be contracted in $E'_i(l-1)$ (i.e. $E'_i$ contains $X_j^d$, but must contain $X_j$ in $E'_i(l-1)$).
This requires at least $d(t-2)$ contractions.
Moreover, $B$ contains the $B_1$ substring, whereas $E'_i$ contains $B_1^*$.  There must have been at least $dn + 2d - 1$ affecting the $\Z_i^{01}$ substring of $E'_i$.  %We note that all contractions enumerated so far are distinct, as otherwise one of these contractions would remove all occurrences of a symbol of $E'_i$.  
Counting the contraction removing $E'_i(l-1)$, this implies the existence of $d(t - 2) + dn + 2d - 1 + 1 = dn + dt$ contractions affecting $E'_i$.  

If instead $\hat{E}'_i(l - 1)$ was located in a $\Z_{2p}^1\X\L$-block, call this block $P$, then it suffices to note that $P$ contains $B_0$ as a substring whereas $E'_i$ contains $B_0^*$.  Counting the contraction that removes $E'_i(l-1)$, it follows that at least $dc + 2d - 1$ contractions must have affected $E'_i$.
%\qed
\end{proof}

The above shows that there are two types of contractions that can remove $E'_i$ from $T(l)$.  Either it uses the $BX\L$ substring at the start of $T(l-1)$, or it uses a block from a $\Z_{2p}^1\X\L$-cluster.
We will call the $E'_i$'s that get removed in the first manner  Type 1, and those that get removed in the second manner Type 2.  
%We note that by Claim~\ref{claim:no-two-eis}, if $E'_i$ and $E'_j$ are Type 1 and $i < j$, then $E'_i$ gets removed before $E'_j$ (otherwise, the contraction removing $E'_j$ would also affect $E'_i$.

%We can now show that our contraction sequence essentially chooses some $X_i$'s to activate first, then proceeds to remove the Type 1 $E'_i$'s.  
We would like to show that every Type 1 $E'_j$ gets removed with the same set of activated $X_i$'s, but it might not be the case.  
Rather, our next goal is to show that ``many'' $E'_j$'s of Type 1 use the same activated $X_i$'s.
For $k \in [p]$, denote by $act(E'_k)$ the set of activated $X_i$'s when $E'_k$ gets removed (i.e. when $E'_k$ is not removed from $T(l-1)$ but is removed from $T(l)$).
Let us partition $[p]$ into intervals of integers $P_a = [1 + am .. m + am]$, where $a \in \{0, \ldots, p/m - 1\}$.  We say that interval $P_a$ is \emph{homogeneous} 
if, for each $i, j \in P_a$ such that $E'_i$ and $E'_j$ are of Type 1, 
$act(E'_i) = act(E'_j)$.  In other words, $P_a$ is homogeneous if all the Type 1 $E'_i$ substrings 
corresponding to those in $P_a$ are removed with the same set of activated $X_i$'s.

\begin{nclaim}
There are at least $p/m - 2n$ homogeneous intervals.
\end{nclaim}

\begin{proof}
Observe that once an $X_i$ is activated, it remains so for the rest of the contraction sequence.  Since there are $n$ of the $X_i$'s, there are only $n + 1$ possible values for $act(E'_k)$ (counting the case when none of them are activated).  There are $p/m$ intervals, and it follows that at most
$n + 1 \leq 2n$ of them are not homogeneous.
%\qed
\end{proof}

%\begin{nclaim}\label{claim:the-act}
%There exists a subset $X^* \subseteq \{X_1, \ldots, X_n\}$ 
%such that at least $p/n$ of the $E'_k$ substrings satisfy $act(E'_k) = X^*$.
%\end{nclaim}

%\begin{proof}
%Observe that once an $X_i$ is activated, it remains so for the rest of the %contraction sequence.  Therefore, there are only $n$ possible values for $act(E'_k)$, $k \in [p]$.  Since there are $p$ possible $E'_k$ values, by the pigeonhole principle, there must be one of the $act(E'_k)$ subsets for which there are at least $p/n$ other $E'_{k'}$
%substrings satisfying $act(E'_k) = act(E'_{k'})$.
%\qed
%\end{proof}

We can now go on with the final elements of the proof.  
Define $cost(E'_i)$ as the number of contractions that affect $E'_i$.
Let $P_{a_1}, \ldots, P_{a_h}$ be the set of homogeneous intervals, $h \geq p/m - 2n$.
Choose the $P_a$ interval among those whose sum of corresponding $E'_i$ costs is minimized --- in other words choose $P_a$ such that 

$$\sum_{i \in P_a} cost(E'_i) = \min_{j \in [h]} \sum_{i \in P_{a_j}} cost(E'_i)$$

By Claim~\ref{claim:no-two-eis}, no two $E'_i$'s share their cost, and by the minimality of $P_a$ the total number of contractions is at least 

$$\left(\frac{p}{m} - 2n\right)\sum_{i \in P_a}cost(E'_i)$$

We will only bother with these contractions and we make no assumption on the non-homogeneous intervals.  Assume that there is at least one $i \in P_a$ such that $E'_i$ is of Type 1.  Then by Claim~\ref{claim:two-removal-ways}, $cost(E'_i)$ is either at least $\min \{dc + dn - 1, dt + dn \}$ where $t = |act(E'_i)|$, or $cost(E'_i)$ is at least $dc + dn - 1$.  If $dt + dn \geq dc + dn - 1$, we may assume that $E'_i$ is of Type 2 since removing $E'_i$ using Type 2 contractions will not increase its cost.  We will therefore assume that if there is at least one $E'_i$  of Type 1 in $P_a$, then $dt + dn < dc + dn - 1$ and thus $cost(E'_i) \geq dt + dn$.

Now, choose any $i$ in $P_a$ such that $E'_i$ is of Type 1, and let $W$ be the set of vertices of $G$ corresponding to those in $act(E'_i)$.
That is, $v_j \in W$ if and only if $X_j$ is activated when $E'_i$ gets removed.
If there does not exist an $E'_i$ of Type 1 to choose, then define $W = \emptyset$.
Denote $|W| = t$ and $|E(W)| = s$.
We claim that $W$ is a subgraph of $G$ satisfying $c(m - s) + ts \leq r$.

Assume $c(m - s) + ts > r$ (otherwise, we are done).  As we are dealing with integers, this means $c(m - s) + ts \geq r + 1$.  We will derive a contradiction on the assumed number of contractions.
For any $E'_i$ where $i \in P_a$, by Claim~\ref{claim:two-removal-ways}, either $e_i$ is not in $W$ and $cost(E'_i) \geq dc + dn - 1$, or
$e_i$ is in $W$ and $cost(E'_i) \geq dt + dn$. 
Note that we needed to choose $P_a$ to be homogeneous to guarantee that every Type 1 $E'_i$ uses the same value of $t$ in the cost $dt + dn$.
It follows that the total number of contractions is at least
%\begin{align*}
%&  \left( \frac{p}{nm} \right) [(m - s)(dc + dn - 1) + s(dt + dn)] \\
%=& \left( \frac{p}{nm} \right)[d((c + n)(m - s) + s(t+n)) - m + s] \\
%=&\left( \frac{p}{nm} \right)[d(c(m - s) + st) + d(nm - ns + ns) - m + s] \\
%=&\left( \frac{p}{nm} \right) \cdot d \cdot [c(m - s) + st + nm] + \left( \frac{p}{nm} \right)(s - m) \\
%\geq &\left( \frac{p}{nm} \right) \cdot d \cdot [r + 1 + nm] + \left( \frac{p}{nm} \right)(s - m) \\
%=& \left( \frac{p}{nm} \right) \cdot d \cdot [r + nm] + \left( \frac{p}{nm} \right)(d + s - m) \\
%=& \frac{p}{m} \cdot d(r + nm) - \left( \frac{p}{nm} \right)d(n - 1)(r + nm) + \left( \frac{p}{nm} \right)(d + s - m)
%\end{align*}
\begin{align*}
&\left( \frac{p}{m} - 2n\right)[(m - s)(dc + dn - 1) + s(dt + dn)] \\
=& \left( \frac{p}{m} - 2n\right)[d((c + n)(m - s) + s(t+n)) - m + s] \\
=&\left( \frac{p}{m} - 2n\right)[d(c(m - s) + st) + d(nm - ns + ns) - m + s] \\
=&\left( \frac{p}{m} - 2n\right) \cdot d \cdot [c(m - s) + st + nm] + \left( \frac{p}{m} - 2n\right)(s - m) \\
\geq &\left( \frac{p}{m} - 2n\right) \cdot d \cdot [r + 1 + nm] + \left( \frac{p}{m} - 2n\right)(s - m) \\
=& \left( \frac{p}{m} - 2n\right) \cdot d \cdot [r + nm] + \left( \frac{p}{m} - 2n\right)(d + s - m) \\
=& \frac{p}{m} \cdot d(r + nm) - 2dn(r+nm) + \left( \frac{p}{m} - 2n\right)(d + s - m)
\end{align*}

Now if $d$ and $p$ are large enough, the above is strictly greater $p/m \cdot d (r + nm) + 4cdn$, leading to a contradiction.  
Our chosen values $d = m + 1$ and $p = (n + m)^{10}$ easily verify this.  We have therefore shown that $W$ has the desired cost.  This concludes the proof.
%\qed
\end{proof}

\end{document}